\documentclass[11pt]{article}
\usepackage{url,geometry}
\usepackage{graphics,latexsym,amsfonts,verbatim,amsmath}
\usepackage{algorithm}
\usepackage{algpseudocode}
\usepackage{todonotes}
\usepackage{lineno}
\usepackage{multirow}
\usepackage{tikz}
\usepackage{pbox}
\usepackage{bibentry} 
\usetikzlibrary{shapes,arrows}
\usepackage{authblk}
\usepackage{setspace}
\usepackage{algpseudocode}

\usepackage{booktabs}
\usepackage{paralist}
\usepackage{colortbl}
\usepackage{xcolor}
\usepackage[bookmarksnumbered=true]{hyperref}
\usepackage{natbib}

\def\hat{\widehat}

\def\bm{\mathbf}

\newcommand{\exclude}[1]{}
\newtheorem{lemma}{Lemma}
\newtheorem{alg}{Algorithm}

\newtheorem{proposition}{Proposition}

\tikzstyle{box} = [rectangle, rounded corners, minimum width=1cm, minimum height=0.5cm, text centered, text width=3cm, draw=black]
\tikzstyle{box2} = [rectangle, rounded corners, minimum width=1cm, minimum height=0.5cm, text centered, text width=2cm, draw=black]

\tikzstyle{arrow} = [thick,->,>=stealth]

\title{On the Asymptotics of Graph Cut Objectives for Experimental Designs of Network A/B Testing}
\author{Qiong Zhang\thanks{qiongz@clemson.edu}}
\affil{School of Mathematical and Statistical Sciences, Clemson University, SC, USA}

\begin{document}
\maketitle

\begin{abstract}
A/B testing is an effective way to assess the potential impacts of two treatments. For A/B tests conducted by IT companies, the test users of A/B testing are often connected and form a social network. The responses of A/B testing can be related to the network connection of test users. This paper discusses the relationship between the design criteria of network A/B testing and graph cut objectives. We develop asymptotic distributions of graph cut objectives to enable rerandomization algorithms for the design of network A/B testing under two scenarios. 
\end{abstract}


\section{Introduction}

IT companies such as Facebook and LinkedIn frequently conduct controlled experiments to evaluate the performance of two versions (e.g., A and B) of products, features, services, etc. This task is an experimental design problem that requires assigning the users to one of the two versions and collecting their responses for evaluation \citep{larsen2023statistical}. Often, the users are connected through the apps and form a social network.  We refer to network A/B testing as the case that the users participating in A/B testing experiments are connected in a social network, and the social connection may imply the potential dependence between connected users. Therefore, network A/B testing design requires assigning users to A or B version according to their network connection with other users,
some examples of recent works including 
\cite{gui2015network, parker2017optimal, basse2018model, pokhilko2019doptimal, zhang2022locally}.   In the literature of graph theory and optimization \citep{ben2001lectures,gross2005graph}, this problem is related to cutting the graph into two partitions with respect to some objectives.
In this section, we first provide an overview of the graph cut problem and then connect it with the experimental objectives of network A/B testing.

\subsection{Graph Cut Problems}

Consider an undirected and unweighted graph with $n$ vertexes, each representing a user in the social network given by the graph. We can express the connection between two vertexes by an $n\times n$ adjacency matrix $W=\{w_{ij}\}$ whose $(i,j)$-th entry is $w_{ij}$. 
The diagonal entries $w_{ii}$'s of this matrix are loaded with zeros, whereas the off-diagonal entries are
\begin{equation}
w_{ij}=\begin{cases}
1, \quad\text{if there is an edge between vertexes $i$ and $j$}\\
0,  \quad\text{otherwise}.
\end{cases}
\end{equation} 
Let $\bm x=(x_1, \ldots, x_n)^\top\in \{-1, 1\}^n$ be the assignments of two options to the $n$ users. This is equivalent to cutting the graph into two disjoint subsets, each with the users assigned to one of the two options, respectively. If there is a cut between two connected users $i$ and $j$ (i.e., $w_{ij}=1$), then $x_i$ and $x_j$ take different values 1 and -1.  

A graph cut is minimum if the edges across two subsets resulting from the cut are minimized \citep{gross2005graph}. Equivalently, two connected users are more likely to receive the same treatment. This problem can be formulated by
\begin{equation}\label{eq:mincut}
    \mathrm{max}_{\bm x\in \{-1, 1\}^n}w_{ij}x_ix_j,~~\mathrm{s.t.}~~-n+1\leq\sum^n_{i=1}x_i\leq n-1.
\end{equation}
By maximizing the objective, the solution to this problem tends to assign the same value to $x_i$ and $x_j$ if $i$
and $j$ are connected. The constraint $-n+1\leq\sum^n_{i=1}x_i\leq n-1$ rules out the situation that all $x_i$'s are assigned with 1 or -1 as a trivial maximum of the objective. The value of $\sum^n_{i=1}x_i$ can be constrained to be zero or in a small interval containing zero if the sizes of the two sub-graphs are required to be relatively the same. The minimum cut problem in \eqref{eq:mincut} is polynomial-time solvable \citep{lawler2001combinatorial}.

A graph cut is maximum if the edges across two subsets are maximized \citep{ben2001lectures}, which is equivalent to 
\begin{equation}\label{eq:maxcut}
    \mathrm{min}_{\bm x\in \{-1, 1\}^n}w_{ij}x_ix_j.
\end{equation}
The solution of this problem tends to assign opposite signs to connected vertexes $i$ and $j$. The graph cut problem is related to experimental designs for network A/B testing. We describe the connections in the following section.

\subsection{Optimal Designs in Network A/B Testing}

Assume that the users' responses are modeled by
\begin{equation}\label{eq:lm}
y_i=\alpha+x_i\beta+\delta_i,
\end{equation}
where $\alpha$ is the intercept, $\beta$ represents the treatment effect, and $\delta_i$ represents the network effect. Next, we describe
two common scenarios of the network effect models.

\paragraph{Scenario I: Network Correlated Responses.} Under this scenario, 
two connected users are assumed to share common features. Thus, the responses of two connected users are correlated due to their common features, examples include \cite{basse2018model, pokhilko2019doptimal, zhang2022locally}. Let $\boldsymbol{\delta}=\{\delta_1, \ldots, \delta_n\}^\top$ with $\delta_i$ be the error term from \eqref{eq:lm}. Following the assumption in 
 \citep{zhang2022locally},
we have that
\begin{equation}\label{eq:mvn}
    \boldsymbol{\delta}\sim \mathcal{MVN}_n(0, \sigma^2 R(W, \rho)^{-1}),
    \end{equation}
with $\sigma^2$ being the variance parameter, $\rho$ being the correlation parameter that characterizes the strength of the correlation between responses of connected users. 
An example of $R(W, \rho)$ in \eqref{eq:mvn} is the conditional auto-regressive model \citep{besag1974spatial} with
\begin{equation}\label{eq:car}
R(W,\rho)=\mathrm{diag}\left(d_1, \ldots, d_n\right)-\rho W
\end{equation}
where $d_i=\sum^n_{j=1}w_{ij}$ is the degree of the $i$-th vertex  for $i=1,\ldots, n$.
According to \cite{pokhilko2019doptimal}, given the value of $\rho$,  the variance of the weighted least squared estimator of the treatment effect $\hat\beta$ is 
\begin{equation}\label{eq:obj1}
\mathrm{Var}\left(\hat\beta\right)=\sigma^2\left(
\sum_{i,j}w_{i,j}-\rho\sum_{i,j}w_{i,j}x_ix_j
-\frac{(1-\rho)\left(\sum^n_{i=1} d_i x_i\right)^2}{\sum_{i,j}w_{i,j}}\right)^{-1}.  
\end{equation}
For a given network $W$, a lower bound of \eqref{eq:obj1} can be expressed by
\begin{equation}\label{eq:ideal1}
\mathrm{Var}\left(\hat\beta\right)\geq \frac{\sigma^2}{(1+\rho)}\left\{\sum_{i,j}w_{i,j}\right\}^{-1}.
\end{equation}
As noted by \cite{pokhilko2019doptimal}, this lower bound is attained if  $\sum_{i,j} w_{ij} x_i x_j=-\sum_{i,j} w_{ij}$ and 
$\sum^n_{i=1} d_i x_i=0$ hold exactly. The first condition 
$\sum_{i,j} w_{ij} x_i x_j=-\sum_{i,j} w_{ij}$ requires that any pair of connected users are assigned with different treatments, whereas the second condition $\sum^n_{i=1} d_i x_i=0$ requires that the treatment allocation is stratified with respect to the degrees of the vertexes.
It is obvious that this lower bound can not be exactly attained for all the networks. 
To reduce the variance of the estimated treatment effect, it is desired to allocate design points to produce smaller values of $\sum_{i,j} w_{ij} x_i x_j$
and $ \left(\sum^n_{i=1} d_i x_i\right)^2$. Also, $|\sum^n_{i=1}x_i|\leq 1$ are required sometimes to ensure that the design is balanced over the two treatments.

\paragraph{Scenario II: Network Interference.} Under this scenario, the users' responses are affected by the design allocation of their connected users. A commonly used model (e.g., \cite{gui2015network,parker2017optimal}) for $\delta_i$ is
\begin{equation}\label{eq:inter}
\delta_i=\sum^n_{j=1}w_{ij}\alpha+\left(\sum^n_{j=1} w_{ij}x_j\right)\gamma +\varepsilon_i,
\end{equation}
where $\alpha$ and $\gamma$ are unknown parameters, and $\varepsilon_i$ is an independent random error with mean zero and variance $\sigma^2$. Under this model assumption, the variance of the least squared estimator of $\beta$ is 
\begin{equation}\label{eq:obj2}
\mathrm{Var}\left(\hat{ \beta}\right)=\sigma^2\left\{{\bm x}^\top\left(I-F_n(F^\top_n F_n)^{-1}F^\top_n\right) {\bm x}\right\}^{-1}
\end{equation}
where $ F_n=\left[
\bm 1_n ~ W\bm 1_n ~ W{\bm x}
 \right]$
is an $n\times 3$ matrix and $\bm 1_n$ is an $n$ dimensional vector loaded with ones. It is obvious that 
\begin{equation}\label{eq:ideal2}
\mathrm{Var}\left(\hat{ \beta}\right)\geq \frac{\sigma^2}{n}.
\end{equation}
The objective $\mathrm{Var}\left(\hat{ \beta}\right)$ is minimized if 
${\bm x}^\top F_n(F^\top_n F_n)^{-1}F^\top_n {\bm x}=0$. Then, if $n$ is even,  the sufficient condition to minimize $\mathrm{Var}\left(\hat{ \beta}\right)$ is that
\[
\sum^n_{i=1}x_i=0,\quad
\sum^n_{i=1}d_{i}x_i=0\quad
\mathrm{and}\quad
\sum^n_{i=1}\sum^n_{j=1}w_{ij}x_i x_j=0
\]
Through the two examples under the two scenarios, we see that the design criteria for network A/B testing often contain three components:
\begin{align}
\bm x^\top W \bm x&=\sum^n_{i=1}\sum^n_{j=1} w_{ij}x_ix_j,\label{eq:g}\\
\bm x^\top W \bm 1_n&=\sum^n_{i=1}d_{i}x_i,\label{eq:f1}\\
\bm x^\top \bm 1_n&=\sum^n_{i=1}x_i.\label{eq:f2}
\end{align}
If $|\bm x^\top \bm 1_n|\leq 1$, the design $\bm x$ is balanced over the two treatments. If $\bm x^\top W \bm 1_n=0$, the design $\bm x$ is balanced with respect to the degrees of the network connection of all $n$ users.
Therefore, $|\bm x^\top \bm 1_n|\leq 1$ and $\bm x^\top W \bm 1_n=0$ give two balanced constraints of the designs. 
The objective $\bm x^\top W \bm x$ is minimized under the scenario of network correlated responses, which is the objective of the max-cut problem in \eqref{eq:maxcut}. As noted by \cite{pokhilko2019doptimal}, the ideal case of the optimal design is given by solving the max-cut problem with the two balanced constraints:
\begin{align*}
\mathrm{min} &\quad \bm x^\top W \bm x\\
\mathrm{s.t.}& \quad -\delta_1\leq\bm x^\top W \bm 1_n\leq \delta_1\\
& \quad -\delta_2\leq \bm x^\top \bm 1_n\leq \delta_2\\
&\quad \bm x\in\{-1, 1\}^n,
\end{align*}
for $\delta_1>0$ and $\delta_2>0$, where $\delta_2$ can be set to be one if an exact balance over the two treatments is required.  
The objective $\bm x^\top W \bm x$ should be close to zero under Scenario II. This problem can not be directly solved as a  
 min-cut or a max-cut problem. A visualization of the ideal cases of the optimal designs is given by Figure \ref{fig: Visualization}. The network contains 24 users, and the users form 12 pairs, with each pair of users connected. The two colors denote the allocation of two treatments. 

 \begin{figure}[!h]
\centering
\includegraphics[scale=0.4]{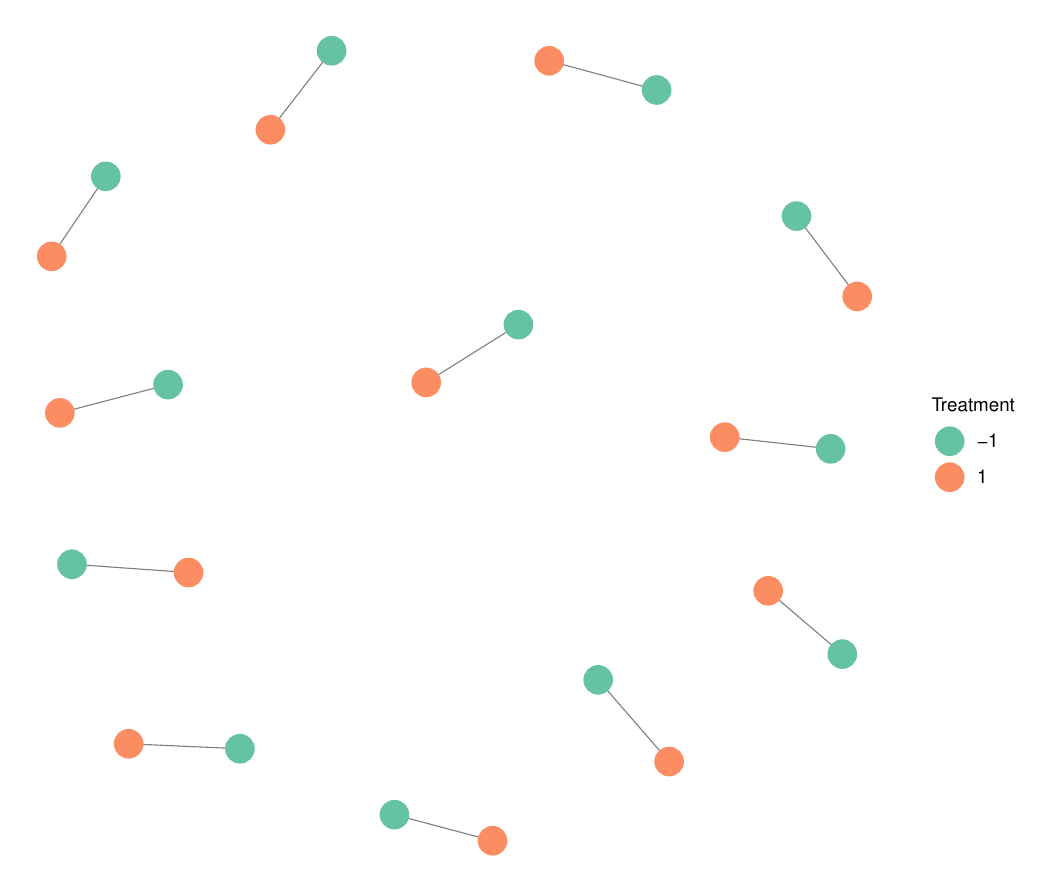}
\includegraphics[scale=0.4]{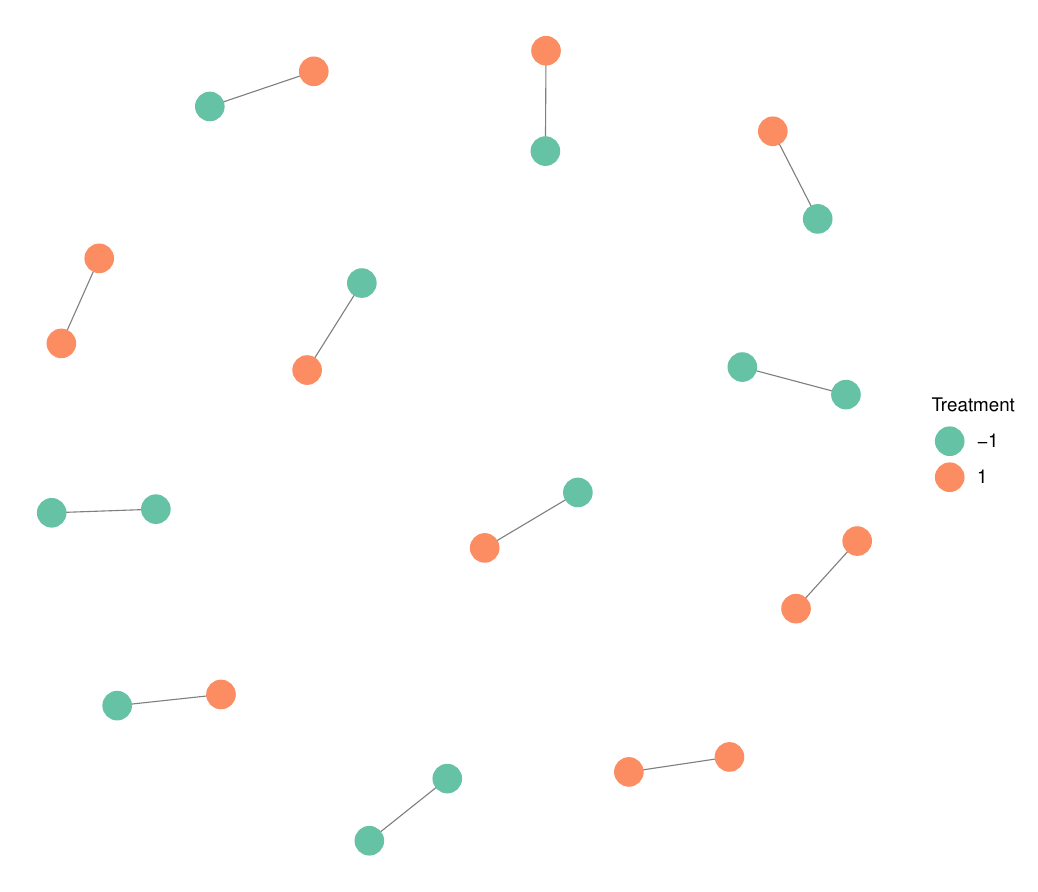}
\caption{The optimal designs of A/B testing under Scenario I (Left) and Scenario II (Right) for a network containing 24 users.}\label{fig: Visualization}
\end{figure}

We observe that the left of Figure \ref{fig: Visualization} shows an optimal design under Scenario I, which allocates different treatments for each pair of connected users to attain the minimized value of the objective in \eqref{eq:g}. The right of Figure \ref{fig: Visualization} shows
 an optimal design under Scenario II. The resulting design contains three pairs of users allocated with treatment 1, three pairs allocated with treatment -1, and six allocated with different treatments. Therefore, the value of the objective in \eqref{eq:g} attains zero exactly. For an arbitrary network, it is not guaranteed that the ideal optimal values in \eqref{eq:ideal1} and
 \eqref{eq:ideal2} can be attained. Also, obtaining the exact optimal design for large social networks without randomization can be inefficient in computation. Also, exact optimal design without sufficient randomization can cause robustness concerns in statistical inference (e.g., \cite{morgan2012rerandomization}).
Next, we propose an algorithm to obtain random designs
that reduces the variance of estimated treatment effects under each scenario.

\section{A Random Design Algorithm for Network A/B Testing}

Random designs with a certain amount of variance reduction can often be obtained by rerandomization until satisfying some stopping criteria (e.g., \cite{morgan2012rerandomization, li2017general}). Let
\begin{equation}\label{eq:gcon}
g(\bm x; W, \delta_1, \delta_2)=\begin{cases}
1 \quad \mathrm{if} \quad -\delta_1\leq\bm x^\top W \bm 1_n\leq \delta_1 \quad \mathrm{and}\quad-\delta_2\leq \bm x^\top \bm 1_n\leq \delta_2\\
0 \quad \mathrm{o.w.}
\end{cases},
\end{equation}
for $\delta_1>0$ and $\delta_2>0$. The value of this function indicates whether or not the two balanced constraints are met. 
For Scenarios I \& II, 
we denote the stopping rule associated with the graph cut objectives by 
\begin{equation}\label{eq:phi}
 \phi_1(\bm x; W, c)=
 \begin{cases}
 1 \quad \mathrm{if} \quad \bm x^\top W \bm x\leq c \\
 0 \quad \mathrm{o.w.}
 \end{cases}
 \quad\mathrm{and}\quad
 \phi_2(\bm x; W, c)=
 \begin{cases}
 1 \quad \mathrm{if} \quad |\bm x^\top W \bm x|\leq c \\
 0 \quad \mathrm{o.w.}
 \end{cases},
\end{equation}
respectively.

For given values of $\delta_1$, $\delta_2$, and $c$, we can obtain the random design and check the values of 
$g(\bm x; W, \delta_1, \delta_2)$ and $\phi_1(\bm x; W, c)$ (or $\phi_2(\bm x; W, c)$) in a sequence to form a rerandomization algorithm \citep{pokhilko2019statistical}.
For a stopping rule containing two or more criteria, it can be challenging to investigate the efficiency of the algorithm and control the running time.
Therefore, we first propose an algorithm that can generate
random designs $\bm x$  under the condition that $g(\bm x; W, \delta_1, \delta_2)=1$ for some $\delta_1$ and $\delta_2$.

\begin{alg}\label{alg:rand}
We obtain a random design $\bm x$ of size $n$ for a given network $W$ with degrees $d_i$'s. 
\begin{itemize}
\item[\bf Step 1:]  Define 
\begin{equation}\label{eq:randomize}
\tilde d_i=d_i+u_i,\quad\mathrm{with}\quad u_i\sim\mathrm{i.i.d.}\quad U(0, 1).
\end{equation} 
Denote the rank of $\tilde{d}_i$ by $r_i$, which is taking value from 1 to $n$. 
\item[\bf Step 2:] Denote
\[
c_i=\begin{cases}
\left\lfloor \frac{r_i}{2}\right\rfloor &  \mathrm{if}\quad n\quad\mathrm{is~even}\\
\left\lfloor \frac{r_i-1}{2}\right\rfloor
& \mathrm{if}\quad n\quad\mathrm{is~odd}
\end{cases},
\]
where $\lfloor r\rfloor$ takes the largest integer that is smaller or equal to $r$.
Then the vertexes are divided into $n/2$ or $(n+1)/2$ groups, each sharing a common value of $c_i$ and with size one or two. 
\item[\bf Step 3:] For the groups given by Step 2, we randomly shuffle  $\{-1, -1\}$ within each group independently to assign their corresponding design values for the vertexes with $c_i>0$. If $n$ is odd, there will be one group containing a single vertex with $c_i=0$. We randomly assigned with 1 or -1 to the user associated with this vertex. 
\end{itemize}

\end{alg}
The purpose of step 1 is to make sure that there are no ties in the rank of the vertexes, and in the meantime, the random variables $u_i$'s in step 1 provide extra randomness on the resulting design $\bm x$.  
The aim of step 2 and step 3 is to balance over the degrees of different vertexes. 
As a result of utilizing this algorithm, $|\bm x^\top \bm 1_n|$ will be set as zero or one exactly and $\bm x^\top W \bm 1_n$ will be set close to zero by generating a vector $\bm x$ that balancing the degrees. We state Proposition \ref{prop:balance} below to formally demonstrate that the random design $\bm x$ given by Algorithm \ref{alg:rand} satisfies the balanced constraints in \eqref{eq:gcon} for some $\delta_1$ and $\delta_2$. 

\begin{proposition}\label{prop:balance}
Algorithm \ref{alg:rand} leads  a random design $\bm x$ that satisfying $g(\bm x; W, \delta_1, \delta_2)=1$ for any $\delta_1\geq c(W)$ and any $\delta_2\geq 1$, where
\[
c(W)=\begin{cases}
\sum^{n/2}_{i=1}\left(d_{(2i+1)}-d_{(2i)}\right) & \mathrm{if}\quad n\quad\mathrm{is~even}\\
d_{(1)}+\sum^{(n-1)/2}_{i=1}\left(d_{(2i+1)}-d_{(2i)}\right)
& \mathrm{if}\quad n\quad\mathrm{is~odd}
\end{cases},
\]
with $d_{(i)}$'s be the ordered degrees based on the rank $r_i$'s in \eqref{eq:randomize}.
The value of $c(W)$ is a constant given the network adjacency matrix $W$.
\end{proposition}

The random design $\bm x$ given by Algorithm \ref{alg:rand} meet the balanced criteria in $g(\bm x; W, \delta_1, \delta_2)$. By utilizing this algorithm, we can rerandomize designs $\bm x$ to obtain a random design satisfying one of the stopping rules in \eqref{eq:phi} for a given $c$. The proposed random design algorithm is described below. 

\begin{alg}\label{alg:design}
Let $T$ be the maximum number of randomizations. For $t\leq T$, we loop over the following two steps until the stopping rule in the second step is met. 
\begin{itemize}
\item[\bf Step 1:] Generate a random design $\bm x_t$ using Algorithm \ref{alg:rand}.
\item[\bf Step 2:] Compute $\bm x^\top_t W\bm x_t$. For Scenario I (or II),  stop the loop if $\phi_1(\bm x_t; W, c )=1$ (or $\phi_2(\bm x_t; W, c)=1$ for Scenario II) is satisfied. 
\end{itemize}
For $t<T$, return the design $\bm x_t$. For $t=T$, return the design  $\bm x_{t^\ast}$ with 
\[
t^\ast=\begin{cases}
\mathrm{argmin}_{t=1, \ldots, T} \bm x^\top_t W\bm x_t&\mathrm{for~Scenario~I}\\
\mathrm{argmin}_{t=1, \ldots, T} |\bm x^\top_t W\bm x_t|&\mathrm{for~Scenario~II}\\
\end{cases}
\]
\end{alg}

This algorithm is provided for a given threshold value $c$ for \eqref{eq:phi} and some specific values of $\delta_1$ and $\delta_2$ in Proposition \ref{prop:balance}. It is important to understand how small those values are compared with the distributions of the objectives in  \eqref{eq:f2} under random designs. 
In the next section, we propose some asymptotic results to support the investigation of their distributions with random designs.

\section{Asymptotic Results on the Graph Cut Objectives}

Proposition \ref{prop:balance} gives the sufficient lower bounds of $\delta_1$ and $\delta_2$ given by Algorithm \ref{alg:rand}. In practice, the value of $\bm x^\top W\bm 1_n$ given by this randomization algorithm can be smaller than $c(W)$. 
Although Proposition \ref{prop:balance} specifies some values of $\delta_1$ and $\delta_2$ that meet the balanced constraints in \eqref{eq:gcon}, it is also necessary to justify how small those values are compared with the results from complete random designs. First of all, the value of $|\bm x^\top \bm 1_n|$ given by Algorithm \ref{alg:rand} is taking the minimum possible value. Given that  $|\bm x^\top \bm 1_n|\leq 1$, we now develop the asymptotic distribution of $\bm x^\top W\bm 1_n$ to compare the lower bound of $\delta_1$ for justification purpose.

\begin{proposition}\label{prop:asy1}
    For a random allocation of $\bm x$ satisfying $|\bm x^\top \bm 1_n|\leq 1$, and   \[ 
   \frac{\mathrm{max}_{1\leq i\leq n}(d_i-\bar d)^2}{\sum^n_{i=1}(d_i-\bar d)^2}\rightarrow 0,\quad\mathrm{with}\quad \bar d=n^{-1}\sum^n_{i=1}d_i
   \]
   as $n\rightarrow\infty$,
    we have that
    \[
    \frac{\sum^n_{i=1} d_i x_i}{\sqrt{\sum^n_{i=1} (d_i-\bar d)^2}}\rightarrow N(0, 1)
    \]
    in distribution as $n\rightarrow\infty$.
\end{proposition} 
The asymptotic distribution is used to compute the probability of more extreme cases compared to a given threshold value $c$
\begin{equation}\label{eq:prob}
\mathrm{P}\left(\left|\bm x^\top W\bm 1_n\right|\leq c\Bigg\rvert|\bm x^\top\bm 1_n|\leq 1\right)\approx 2\Phi\left(\frac{c}{\sqrt{\sum^n_{i=1} (d_i-\bar d)^2}}\right)-1
\end{equation}
where $\Phi(\cdot)$ is the CDF of the standard normal distribution. Once the network adjacency matrix $W$ is given, we are able to compute this probability. For example, by specifying $c=c(W)$, the above probability tells the possibility of a random design $\bm x$ satisfying $|\bm x^\top W \bm 1_n|\leq c(W)$ given that $|\bm x^\top \bm 1_n|\leq 1$. Also, given a specific design $\bm x_0$,  we can specify $c=\bm x^\top_0 W\bm 1_n$ to validate if this design can lead to a sufficiently small value.

For both scenarios, the value of $c$ can be chosen as a smaller quantile according to the distributions of the graph cut objective
$\bm x^\top W \bm x$ or $|\bm x^\top W \bm x|$ 
for a random design $\bm x$ generated from Algorithm \ref{alg:rand}. We first define some convenient notation and then provide the asymptotic distribution of $\bm x^\top W \bm x$. 
Let $\tilde W$ be the adjacency matrix reorder by $r_i$ in \eqref{eq:randomize} from the smallest to the largest. If $n$ is an odd number, $r_i=1$ is removed before reordering. 
Then the size of $\tilde W$ is even. 
We further define
\begin{equation}\label{eq:W0}
W_0=\left(\bm I \otimes [1,-1]\right)\tilde W \left(\bm I \otimes [1,-1]^\top\right)
\end{equation}
with $\bm I$ be the identity matrix with size 
$n/2$ or $(n-1)/2$ for $n$ being even or odd.

\begin{proposition} \label{prop:asy2}
Given $W_0$ in \eqref{eq:W0}, we assume that
\[
\frac{\mathrm{min}_{i=1,\ldots, n}d_i}{\sum_{i\neq j}w^2_{0,ij}}\rightarrow 0
\quad\mathrm{and}\quad
\frac{\lambda_{\mathrm{max}}(\tilde W_0)}{\sqrt{\sum_{i\neq j}w^2_{0,ij}}}\rightarrow 0
\quad\mathrm{as}\quad n\rightarrow\infty,
 \]
where $\tilde W_0$ is a matrix with off-diagonal entries equal to the corresponding entries of $W_0$ in \eqref{eq:W0}, and diagonal entries equal to 0 and $w_{0,ij}$ is the $i,j$-th entry
of $W_0$ in \eqref{eq:W0}. The notation $\lambda_{\mathrm{max}}(A)$ gives the maximum eigenvalue of a symmetric matrix $A$. 
Then we have that
\[
\frac{\bm x^\top W\bm x-\mathrm{trace}(W_0)}{2\sqrt{\sum_{i<j} w^2_{0,ij}}}\rightarrow N(0,1)
\]
in distribution as $n\rightarrow\infty$.
\end{proposition}

Under this proposition, we can set $c$ in Algorithm \ref{alg:design} based on the asymptotic normal  distribution 
\[\bm x^\top W\bm x\sim AN\left(\mathrm{trace}(W_0), 4\sum_{i<j} w^2_{0,ij}\right)
\]
For $\phi_1$, we set the stopping threshold as the $\alpha$-th quantile of $\bm x^\top W\bm x$, i.e.
\[
c=\mathrm{trace}(W_0)+2\sqrt{\sum_{i<j} w^2_{0,ij}}\Phi^{-1}(\alpha),
\]
where $\Phi^{-1}$ is the inverse function of the standard normal cumulative distribution function.  
For $\phi_2$, we set the stopping threshold as the $\alpha$-th quantile of $|\bm x^\top W\bm x|$, which asymptotically follows the folded normal distribution with parameters 
$\mathrm{trace}(W_0)$ and $4\sum_{i<j} w^2_{0,ij}$. 
Therefore, the process of rerandomization in Algorithm \ref{alg:design} constructs a Geometric distribution with the success probability approximated by $\alpha$. Then we can set the maximum number of randomization $T$ according to this distribution.  Throughout the numerical results of this paper, we set $T=5000$, $\alpha=0.005$ for Scenario I, and $\alpha=0.1$ for Scenario II.

\section{Numerical Study}\label{sec:num}

We evaluate the performances of the proposed design approach by computing two performance measures. We first generate 1000 random balanced designs satisfying that $\sum^n_{i=1}x_{i}=0$, where $n$ is set to be an even number in the numerical study for convenience. We compute the percentile of $\mathrm{Var}(\hat\beta)$ led by the proposed design approach among the variances from the 1000 random designs: 
\begin{equation}\label{eq:percentile}
\mathrm{Percentile}=\frac{\sum^{1000}_{i=1}I(v_i\leq v_{opt})}{1000},
\end{equation}
where $v_i$'s are the variances of $\hat\beta$ led by the 1000 random balanced designs, whereas $v_{opt}$ is that led by the proposed design approach.  We also compute the optimality gap of the proposed design with respect to the lower bounds (denoted by $v_{lb}$) given by \eqref{eq:ideal1} and \eqref{eq:ideal2}:
\begin{equation}\label{eq:gap}
\mathrm{Gap}=1-\frac{v_{lb}}{v_{opt}}.
\end{equation}
For comparison purposes, we compute the optimality gap of the median variance of the 1000 random balanced designs: $\mathrm{Gap_{median}}=1-v_{lb}/v_{median}$ with
$v_{median}$ be the median variance.

\subsection{Example I: Synthetic Networks}\label{sec:example1}

In this section, we evaluate the performances of the proposed approach using synthetic networks. Given the total number of vertexes $n$ and network density $p$, for $i<j$, $w_{ij}$'s are generated as iid Bernoulli random variables with the probability equal to one be $p$. We remove the isolated vertexes, so the actual size of the generated network can be smaller than $n$. For convenience of implementation, we force the resulting size of the network to be even. 

We first check the probability in \eqref{eq:prob} for the synthetic networks.
The results are depicted in Figure \ref{fg:prob}. In the left of Figure \ref{fg:prob}, we set $c$ in \eqref{eq:prob} be $c(W)$. The results show that the probability of the upper bound is at most 0.1 for networks with a size above 1000, but for networks with a size 100 and smaller density, the probability might be as high as 0.5. For networks with size 1000 or above, the guaranteed upper bound $c(W)$ for $\sum^n_{i=1}d_ix_i$ is small with respect to its asymptotic distribution.  In the right of Figure \ref{fg:prob}, we set $c$
in \eqref{eq:prob} be $\sum^n_{i=1}d_ix_i$  with 100 copies of the design $\bm x$ generated by Algorithm \ref{alg:rand}. The average probability values are reported. The actual values of $\sum^n_{i=1}d_ix_i$ can be significantly smaller than its upper bound $c(W)$ with small probability values compared to their asymptotic distributions.

\begin{figure}[!h]
\centering
\includegraphics[scale=0.75]{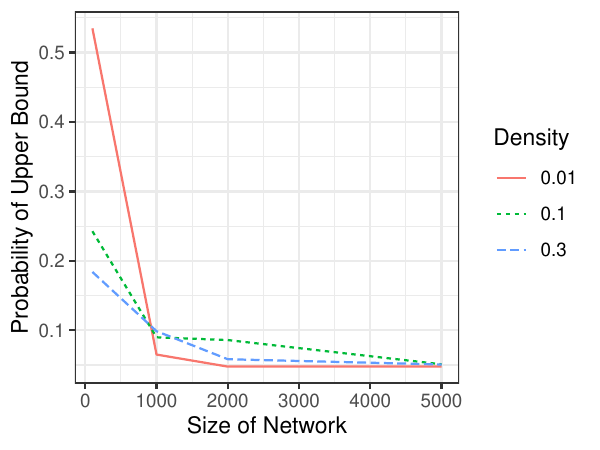}
\includegraphics[scale=0.75]{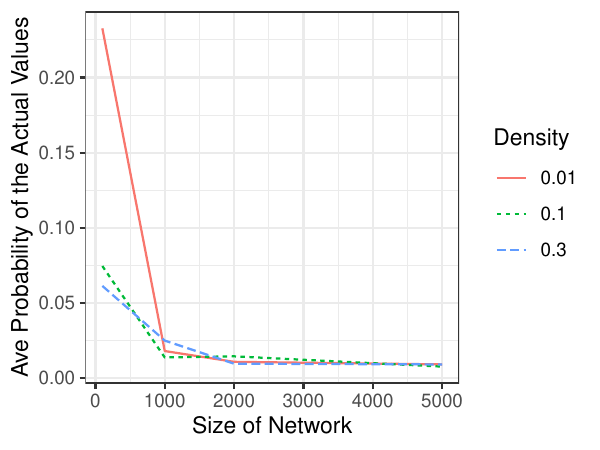}

\caption{The resulting probability in \eqref{eq:prob} with $c=c(W)$ (left, probability of upper bound) and and the average probability of $c=|\sum^n_{i=1}d_ix_i|$ (right, ave probability of the actual values) for 100 random designs generated by Algorithm \ref{alg:rand} with synthetic networks.}\label{fg:prob}
\end{figure}
 
Next, we compute the Percentile and Gap in \eqref{eq:percentile} and \eqref{eq:gap} for the designs generated by the proposed method. The results of Scenarios I and II are given in Table \ref{tb:s}. We vary the size of network $n$ and network density $p$, and generate ten networks under each setting. In the results, we provide the averages of Percentile and Gap over the ten networks under each setting. For comparison purposes, we also give the average value of $\mathrm{Gap}_\mathrm{median}$ in the tables.
For Scenario I, all the percentiles are below 0.004, which shows that the variance led by the proposed design is nearly optimal among the 1000 random designs. And the Gap values of the proposed design are all smaller than $\mathrm{Gap}_{\mathrm{median}}$. The advantage of the proposed design in terms of Gap becomes smaller for $n$ larger than 1000. For Scenario II, the percentiles of the proposed design approach are around 0.01 for smaller $n$ and drop to around 0.005 for larger $n$. The Gap values are smaller than in the case of Scenario I. For $n=1000$ and 2000, the Gap values to the ideal designs are nearly zero.

\begin{table}[h]
\centering
\caption{Averages of Percentile, Gap, and $\mathrm{Gap}_\mathrm{median}$ over ten generated networks under each setting}\label{tb:s}
\begin{tabular}{|c|rr|rrr|}
  \hline
 Scenario&$n$ & $p$ & Percentile & Gap & $\mathrm{Gap}_\mathrm{median}$ \\ \hline
 &50  & 0.1 & 0.0021 & 0.2471 & 0.3278 \\ 
 & & 0.3 & 0.0032 & 0.2836 & 0.3271 \\ 
 &100  & 0.1 & 0.0019 & 0.2853 & 0.3302 \\ 
 I & & 0.3 & 0.0027 & 0.3080 & 0.3299 \\ \cline{2-6}
 &1000  & 0.01 & 0.0019 & 0.3191 & 0.3331 \\ 
 & & 0.1 & 0.0024 & 0.3289 & 0.3330 \\ 
 &2000& 0.01 & 0.0028 & 0.3267 & 0.3332 \\ 
 && 0.1 & 0.0033 & 0.3311 & 0.3332 \\\hline\hline
&50  & 0.1 & 0.0120 & 0.0005 & 0.0448 \\ 
&  & 0.3 & 0.0145 & 0.0011 & 0.0466 \\ 
& 100  & 0.1 & 0.0153 & 0.0004 & 0.0211 \\ 
II & & 0.3 & 0.0080 & 0.0003 & 0.0228 \\ \cline{2-6}
& 1000  & 0.01 & 0.0046 & 0.0000 & 0.0020 \\ 
& & 0.1 & 0.0053 & 0.0000 & 0.0021 \\ 
& 2000  & 0.01 & 0.0037 & 0.0000 & 0.0010 \\ 
&& 0.1 & 0.0051 & 0.0000 & 0.0010 \\ 
   \hline
\end{tabular}
\end{table}

\subsection{Example II: Real Networks from Facebook}

We evaluate the performances of the proposed method using real networks from Facebook \citep{leskovec2012learning}.
This data contains ten sampled social networks collected from survey participants using Facebook app. After removing completely isolated users, the sizes of the ten networks range from 52 to 1034 and the network densities range from 0.034 to 0.150. Similar to Figure \ref{fg:prob}, we first check the probability of more extreme situations given by the upper bound $c(W)$ and the actual values of $|\sum^n_{i=1}d_ix_i|$ for 100 random designs from Algorithm \ref{alg:rand} in Figure \ref{fg:prob_case}. The figure shows that the resulting design given by this Algorithm can achieve smaller values of  $|\sum^n_{i=1}d_ix_i|$ to balance over the degrees of vertexes.

\begin{figure}[!h]
\centering
\includegraphics[scale=0.9]{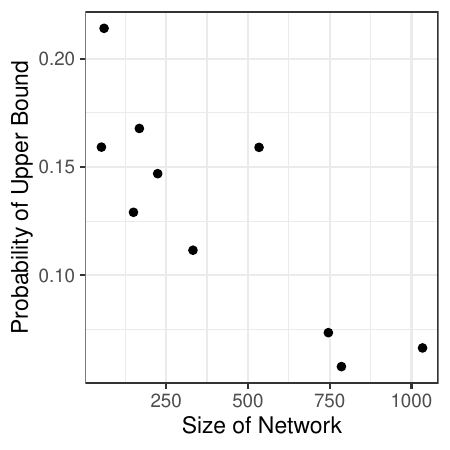}
\includegraphics[scale=0.9]{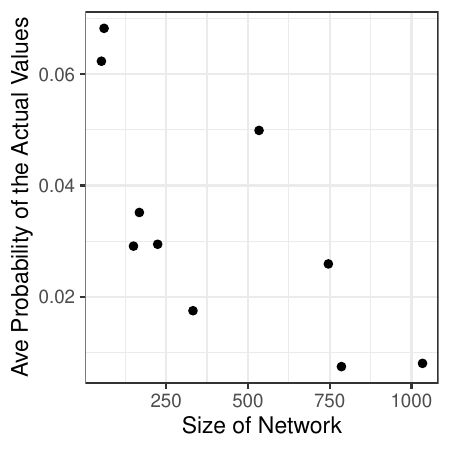}
\caption{The resulting probability in \eqref{eq:prob} with $c=c(W)$ (left, probability of upper bound) and the average probability of $c=|\sum^n_{i=1}d_ix_i|$ (right, ave probability of the actual values) for 100 random designs generated by Algorithm \ref{alg:rand} with the ten sampled networks from Facebook}\label{fg:prob_case}
\end{figure}

For both Scenarios, 
the percentile values of the ten networks are given in Figure \ref{fg:percent_case}, whereas the gap values are given in Figure \ref{fg:gap_case}. The percentiles of Scenario I are all below 0.0025, whereas the percentiles of Scenario II are all below 0.015, which shows that the proposed design leads to smaller variance than the variances from complete random designs. The comparison of Gap and $\mathrm{Gap}_\mathrm{median}$ is similar to the synthetic networks. 

\begin{figure}[!h]
\centering
\includegraphics[scale=0.9]{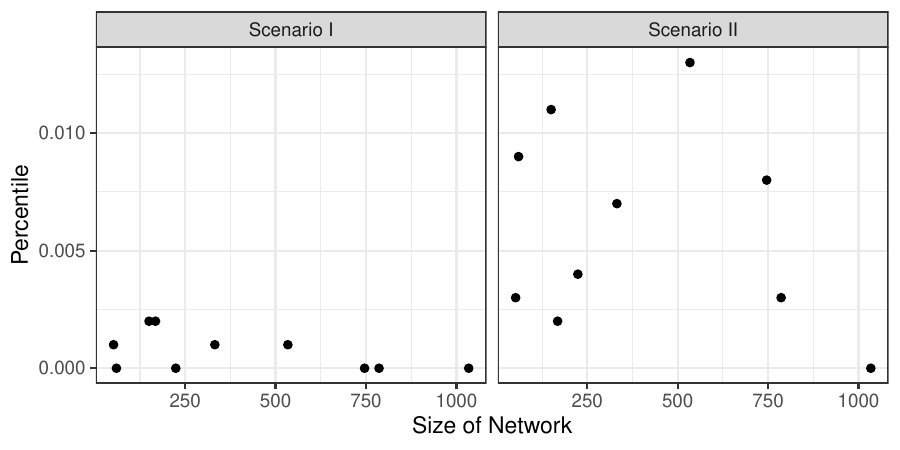}
\caption{The percentile in \eqref{eq:percentile} for the ten sampled networks from Facebook}\label{fg:percent_case}
\end{figure}

\begin{figure}[!h]
\centering
\includegraphics[scale=0.9]{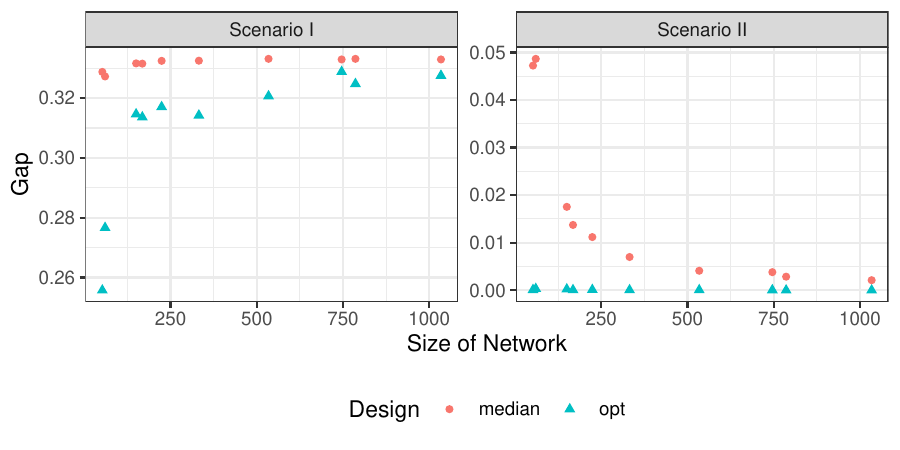}
\caption{The Gap (opt, blue triangles) and $\mathrm{Gap}_\mathrm{median}$ (median, red dots) in \eqref{eq:gap} for the ten sampled networks from Facebook}\label{fg:gap_case}
\end{figure}

\section{Conclusion}
This paper discovered the relationship between the design criteria of network A/B testing and graph cut objectives. We developed asymptotic distributions of two graph cut objectives to enable rerandomization algorithms to design network A/B testing. The numerical results show that the proposed algorithm effectively generates random designs under certain constraints to reduce the variance of parameter estimation from complete random designs. The proposed asymptotic results can also serve as stopping rules of other random algorithms for graph cut-related problems.

\appendix
\section*{Appendix: Proofs and Additional Numerical Validation}

\subsection*{Proof of Proposition \ref{prop:balance}}
Note that $d_{(i)}$'s are the ordered degrees based on the rank $r_i$'s in \eqref{eq:randomize}. Therefore, for $i<i'$, $d_{(i)}\leq d_{(i')}$.
If $n$ is an even number, after applying this algorithm, we have that
\[
\sum^n_{i=1} x_i d_i=\sum^n_{i=1} x_i d_{(i)}=\sum^{n/2}_{i=1}\left[d_{(2i)}x_{2i}-d_{(2i-1)}x_{2i-1}\right],
\]
where $x_{2i}$ and $x_{2i-1}$ are assigned to different treatments, 1 or -1. We can express
\[
\sum^n_{i=1} x_i d_i=\sum^{n/2}_{i=1}\left[d_{(2i)}-d_{(2i-1)}\right]z_i,
\]
with $z_i\in \{-1, 1\}$. Then,
we have that
\[
-\sum^{n/2}_{i=1}\left[d_{(2i)}-d_{(2i-1)}\right]\leq\bm x^\top W\bm 1_n\leq \sum^{n/2}_{i=1}\left[d_{(2i)}-d_{(2i-1)}\right]
\]
and
\[
\bm x^\top \bm 1_n=0.
\]
If $n$ is an odd number, we have that
\[
-d_{(1)}-\sum^{(n-1)/2}_{i=1}\left(d_{(2i+1)}-d_{(2i)}\right)\leq\bm x^\top W\bm 1_n\leq d_{(1)}+\sum^{(n-1)/2}_{i=1}\left(d_{(2i+1)}-d_{(2i)}\right)
\]
and
\[
-1\leq \bm x^\top \bm 1_n\leq 1.
\]
We denote
\[
c(W)=\begin{cases}
\sum^{n/2}_{i=1}\left(d_{(2i+1)}-d_{(2i)}\right) & \mathrm{if}\quad n\quad\mathrm{is~even}\\
d_{(1)}+\sum^{(n-1)/2}_{i=1}\left(d_{(2i+1)}-d_{(2i)}\right)
& \mathrm{if}\quad n\quad\mathrm{is~odd}
\end{cases},
\]
then the conclusion holds.

\subsection*{Proof of Proposition \ref{prop:asy1}}
Note that the asymptotic results for $\sum^n_{i=1} d_i x_i$
is described in \cite{pokhilko2019statistical} under random allocations without $|\bm x^\top \bm 1_n|\leq 1$ constraint.
Under $|\bm x^\top \bm 1_n|\leq 1$, we derive the asymptotic distribution based on finite population asympototics \citep{li2017general}. We first state a Lemma to support the proof of Proposition 2.

\begin{lemma}\label{lemma}
Let $z_n$ be a sequence of random variables that converge in distribution to the standard normal distribution. Let $s$ be a binary random variable with 
\[
\mathrm{P}(s=1)=\mathrm{P}(s=-1)=\frac{1}{2},
\]
and $s$ is independent with $z_n$.
Then $g_n=sz_n$ converges in distribution to the standard.
\end{lemma}
\begin{proof}
Let $\Phi(t)$ be the cumulative distribution function (CDF) of $\mathcal N(0,1)$. The CDF of $g_n$ is
\begin{align*}
\mathrm{P}(sz_n\leq t)
&=\mathrm{P}(sz_n\leq t|s=1)\mathrm{P}(s=1)
+\mathrm{P}(sz_n\leq t|s=-1)\mathrm{P}(s=-1)\\
&=\frac{1}{2}\left[\mathrm{P}(z_n\leq t)+\mathrm{P}(z_{n}\geq -t)\right]
\end{align*}
Since $z_n$ converges in distribution to $N(0,1)$,
we have that
\[
\mathrm{P}(z_n\leq t)\rightarrow \Phi(t)\quad\mathrm{and}\quad 
\mathrm{P}(z_{n}\geq -t)\rightarrow 1-\Phi(-t)=\Phi(t)
\]
as $n\rightarrow \infty$. Accordingly, 
$\mathrm{P}(z_{n}\leq t)\rightarrow \Phi(t)$
as $n\rightarrow \infty$. The conclusion holds. 
\end{proof}\\

We now state proof of Proposition \ref{prop:asy1}. We first assume that the $n$ vertexes have been randomly split into two balanced groups with sizes $n_1$ and $n_2$. Without loss of generality, we assume that $n_1=n_2=n/2$ if $n$ is even, whereas $n_1=n_2+1=(n+1)/2$ if $n$ is odd. Let $\bar d_1$ and $\bar d_2$ be the average degrees of the vertexes in the first and second groups, respectively. Then, according to Theorem 1 in \cite{li2017general}, we have that
\[
\frac{\bar d_1-\bar d}{\sqrt{\mathrm{Var}(\bar d_1)}}\rightarrow N(0,1) 
\]
in distribution as $n\rightarrow\infty$
under the assumptions given in this proposition. 
Note that $\mathrm{Var}(\bar d_1)=(n^{-1}_1-n^{-1})(n-1)^{-1}\sum^n_{i=1}(d_i-\bar d)^2$.
Therefore, we have that
\[
\frac{n_1\bar d_1-n_2\bar d_2}{\sqrt{\sum^n_{i=1}(d_i-\bar d)^2}}
=\frac{2n_1\bar d_1-n\bar d}{\sqrt{\sum^n_{i=1}(d_i-\bar d)^2}}=\frac{2n_1}{\sqrt{(n-1)n_1n/n_2}}\frac{\bar d_1-\bar d}{\sqrt{\mathrm{Var}(\bar d_1)}}+\frac{2n_1-n}{\sqrt{(n-1)n_1n/n_2}}\frac{\bar d}{\sqrt{\mathrm{Var}(\bar d_1)}}
\]
As $n\rightarrow\infty$, 
\[
\frac{2n_1}{\sqrt{(n-1)n_1n/n_2}}\rightarrow 1\quad\mathrm{and}\quad \frac{2n_1-n}{\sqrt{(n-1)n_1n/n_2}}\rightarrow 0
\]
Therefore, 
\[
\frac{n_1\bar d_1-n_2\bar d_2}{\sqrt{\sum^n_{i=1}(d_i-\bar d)^2}}\rightarrow N(0, 1)
\]
in distribution as $n\rightarrow\infty$.
A random allocation of the design vector $\bm x$ with the balanced constraint can be equivalently 
implemented by a random split with fixed size $n_1$ and $n_2$ and then randomize 1 and -1 over the two groups. Therefore, the conclusion holds
according to Lemma \ref{lemma}.


\subsection*{Proof of Proposition \ref{prop:asy2}}

 Notice that the entries in $\bm x$ from  Algorithm \ref{alg:rand} are dependent. Then, it is not straightforward to develop the distribution of $\bm x^\top W \bm x$ directly. 
To investigate this, we first provide an alternative expression of the objective $\bm x^\top W\bm x$ to simplify the derivation of the asymptotic distribution of $\bm x^\top W\bm x$ in the following proposition for the case with an even value of $n$.

\begin{lemma} 
Assume that $n$ is even. 
Let $\bm z=(z_1, \ldots, z_{n/2})^\top$ with  $z_{i}\in \{-1,1\}$ for $i=1, \ldots, n/2$.  Denote $\bm x=\bm z\otimes [1, -1]^\top$. 
We have that
\[
\bm x^\top W\bm x=\bm z^\top\left( \bm I_{n/2}\otimes [1,-1]\right)W \left(\bm I_{n/2}\otimes [1,-1]^\top \right)\bm z.
\]
\end{lemma}
The conclusion of this Lemma is obvious, since
$\left(\bm I\otimes [1,-1]^\top \right)\bm z=\bm x$.

For a random design $\bm x$ given by Algorithm \ref{alg:rand}, we reordered by $r_i$s in \eqref{eq:randomize} from the smallest to the largest. 
Denote $\tilde{\bm x}$ be the reordered random design. 
According to the definition of $\tilde W$ and $W_0$ in \eqref{eq:W0} and Algorithm \ref{alg:rand}, 
we have that 
\[
\bm x^\top W\bm x=\tilde{\bm x}^\top \tilde W\tilde{\bm x}=\bm z^\top W_0\bm z,
\]
where $\bm z\in \{-1, 1\}^{n/2}$ is a random vector representing the independent random shuffle of $\{-1, 1\}$ within each group in Step 3 of Algorithm \ref{alg:rand}. Therefore, we can alternatively evaluate the asymptotic distribution of $\bm z^\top W_0\bm z$ with the entries of $\bm z$ being iid random variable. 

We consider $\bm z=(z_1, \ldots, z_{n/2})^\top$ be independent and identically distributed from the distribution with $\mathbb{P}(z_i=1)=\mathbb{P}(z_i=-1)=0.5$. Then
\[
\mathbb{E}(z_i)=0, \quad \mathrm{Var}(z_i)=1, \quad \mathbb{E}|z_i|^3=1. 
\]
Let $A$ be a matrix with $ij$-th element $ a_{ij}=\frac{w_{0,ij}}{\sqrt{\sum_{i\neq j} w^2_{0,ij}}}$ if $i\neq j$, and $a_{ii}=0$ for $i=1,\ldots, n/2$. We have that
\[
\sum_{ij}a^2_{ij}=\frac{\sum_{i\neq j} w^2_{0,ij}}{\sum_{i\neq j} w^2_{0,ij}}=1
\]
According to Theorem 1 in \cite{gotze1999asymptotic}, we have that
\[
\sup_u\left\lvert \mathbb P\left(\frac{\bm z^\top A\bm z}{\sqrt{\mathrm{Var}(\bm z^\top A\bm z})}\leq u\right)-\Phi(u)\right\rvert
\leq C|\lambda_{\mathrm{max}}(A)|,
\]
for some constant $C$. Since 
\[
\lambda_{\mathrm{max}}(A)=\frac{\lambda_{\mathrm{max}}(\tilde W_0)}{\sqrt{\sum_{i\neq j}w^2_{0,ij}}},
\]
 we have that
\[
\frac{\bm x^\top W\bm x-\mathrm{E}\left(\bm z^\top W_0\bm z\right)}{\sqrt{\mathrm{Var}(\bm z^\top W_0\bm z})}=\frac{\bm z^\top W_0\bm z-\mathrm{E}\left(\bm z^\top W_0\bm z\right)}{\sqrt{\mathrm{Var}(\bm z^\top W_0\bm z})}=\frac{\bm z^\top \tilde W_0\bm z}{\sqrt{\mathrm{Var}(\bm z^\top \tilde W_0\bm z})}=\frac{\bm z^\top A\bm z}{\sqrt{\mathrm{Var}(\bm z^\top A\bm z})}\rightarrow N(0,1)
\]
in distribution as $n\rightarrow\infty$.
We see that 
\[
\mathrm{E}\left(\bm z^\top W_0\bm z\right)=\mathrm{trace}(W_0)
\]
and
\[
\mathrm{Var}(\bm z^\top W_0\bm z)=4\sum_{i<j} w^2_{0,ij}
\]
which concludes the case for $n$ being even.

If $n$ is an odd number, the difference between the original objective and the alternative expression comes from the vertex with the minimum degree. 
We have that
\[
\bm x^\top W \bm x=\tilde{\bm x}^\top \tilde W\tilde{\bm x}+2x_{i_1}\sum_{j\neq i_1} w_{i_1, j}x_{j}=\bm z^\top W_0\bm z+2x_{i_1}\sum^n_{j=1} w_{i_1, j}x_{j},
\]
where $i_1$ is the index of the vertex with the smallest $\tilde d_i$ and $x_{i_1}$ is randomly allocated with 1 or -1 independent with $x_j$ for $j\neq i_1$. Let $d_{(1)}=\mathrm{min}_{i=1, \ldots, n}d_i$. Note that
\[
d_{(1)}\left(1-\frac{d_{(1)}}{n-1}\right)\leq\mathrm{Var}\left(x_{i_1}\sum^n_{j=1} w_{i_1, j}x_{j}\right)\leq \sum^n_{j=1}w_{i_1, j}=d_{(1)}.
\]
Then if $d_{(1)}/{\sum_{i<j} w^2_{0,ij}}\rightarrow 0$ holds
\[
\frac{ x_{i_1}\sum^n_{j=1} w_{i_1, j}x_{j}}{\sqrt{\sum_{i<j} w^2_{0,ij}}}\rightarrow 0\quad\mathrm{as}\quad n\rightarrow\infty
\]
in probability. We also have the conclusion holds for $n$ being odd. This concludes the proof. 

\subsection*{Numerical Validation of the Assumptions in Proposition \ref{prop:asy2}}

In this section, we provide numerical validation for the assumptions in Proposition \ref{prop:asy2}. We generate random networks as in Section  \ref{sec:example1}. We vary the values of $n$ (size of network) from 100 to 5000, and the values of network density as 0.01, 0.1, and 0.3. Left of Figure \ref{fg:convergence} shows how $\mathrm{min}_{i=1, \ldots, n}d_i/\sum_{i\neq j}w^2_{0, ij}$ change with the size of network and network density, whereas right of Figure \ref{fg:convergence} shows how $\lambda_{max}(\tilde W_0)/\sqrt{\sum_{i\neq j}w^2_{0, ij}}$ change with the size of network and network density. The trends in the Figures show evidences of the convergences of those values for randomly generated networks as in Section \ref{sec:example1}.

\begin{figure}[!h]
\centering
\includegraphics[scale=0.75]{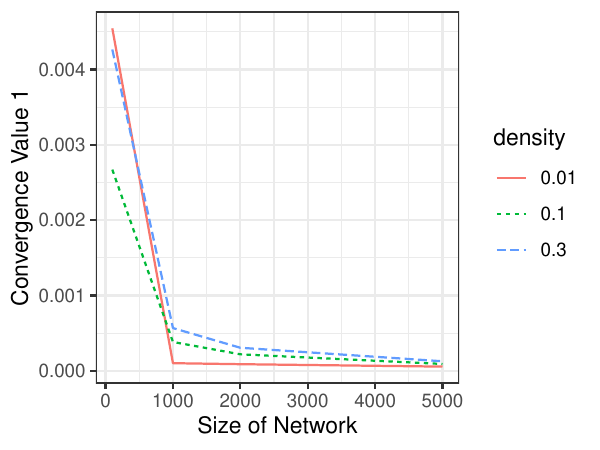}
\includegraphics[scale=0.75]{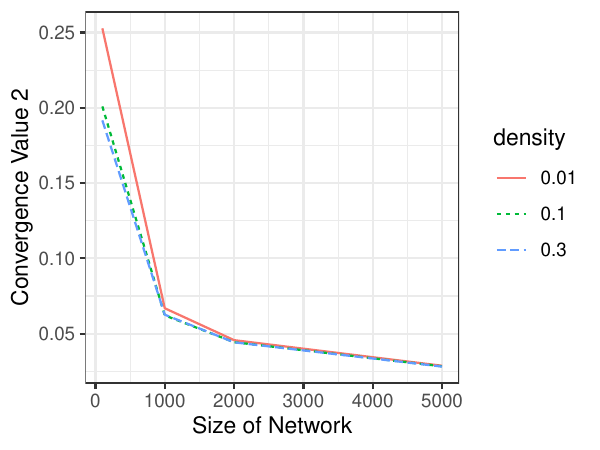}

\caption{The probability of the constraint in \eqref{eq:prob} for generated networks with different sizes and network densities.}\label{fg:convergence}
\end{figure}


\end{document}